\documentclass{IEEEtran}
\usepackage{cite}
\usepackage{multicol}
\usepackage{multirow}
\usepackage{amsthm,amsmath}
\usepackage{graphics}
\DeclareMathOperator*{\argmax}{arg\,max}
\usepackage{amssymb}
\graphicspath{{figures/}}
\usepackage[ruled]{algorithm2e}
\usepackage{color}
\usepackage{mathtools}

\newtheorem{theorem}{Theorem}

\newtheorem*{remark}{Remark}

\usepackage{xpatch}
\makeatletter
\xpatchcmd{\@thm}{\thm@headpunct{.}}{\thm@headpunct{}}{}{}
\begin{document}
\title{Probabilistic Robust Small-Signal Stability Framework using Gaussian Process Learning}
%% To specify the authors when (number of affiliations <= 2)
\author{Parikshit~Pareek,~and~
        Hung~D.~Nguyen$^\star$\vspace{-5pt}
% <-this % stops a space
\thanks{$^\star$Corresponding Author}% <-this % stops a space
\thanks{Authors are with School of Electrical and Electronics Engineering, Nanyang Technological University, Singapore. e-mail: pare0001@ntu.edu.sg, hunghtd@ntu.edu.sg}
\thanks{Parikshit Pareek and Hung Nguyen are supported by NTU SUG, MOE, and NRF fundings.}}%

\maketitle

% The power systems operation depends upon a prediction of load and renewable sources, which has high levels of uncertainty due to varying load patterns and intermittent/distributed generation, etc., translating to changing operating conditions. The prediction errors and uncertainties make small-signal stability assessment challenging, especially with the increased dimension of the parameter space in large-scale systems.
\begin{abstract}
While most power system small-signal stability assessments rely on the reduced Jacobian, which depends non-linearly on the states, uncertain operating points introduce nontrivial hurdles in certifying the system’s stability. In this paper, a novel probabilistic robust small-signal stability (PRS) framework is developed for a power system based on Gaussian process (GP) learning. The proposed PRS assessment provides a robust stability certificate for a state subspace, such as that specified by the error bounds of the state estimation, with a given probability. With such a PRS certificate, all inner points of the concerned subspace will be stable with at least the corresponding confidence level. To this end, the behavior of the critical eigenvalue of the reduced Jacobian with state points in a state subspace is learned using GP. The proposed PRS certificate along with the \textit{Subspace-based  Search} and \textit{Confidence-based Search} mechanisms constitute a holistic framework catering to all scenarios. The proposed framework is a powerful approach to assess the stability under uncertainty because it does not require input uncertainty distributions and other state-specific input-to-output approximations. Further, the critical eigenvalue behavior in a state subspace is analyzed using an upper bound of the eigenvalue variations and their inferences are discussed in detail. The results on three-machine nine-bus WSCC system show that the proposed certificate can find the robust stable state subspace with a given probability.
\end{abstract}


% Use this to place sponsorships
% \thanksto{\noindent $^*$ Corresponding Author.\\ Parikshit Pareek and Hung Nguyen are supported by NTU SUG, MOE, and NRF fundings. }
 
% \thanksto{\noindent $^*$ Corresponding Author.\\ Parikshit Pareek and Hung Nguyen are supported by NTU SUG, MOE, and NRF fundings. \\ Submitted to the 21st Power Systems Computation Conference (PSCC 2020). } 

\section{Introduction}
The small-signal stability is an integral part of power system stability assessment, referring to the system's capacity to withstand small-disturbances and remain in synchronism \cite{kundur1994power, Janusz}. Further, stability assessment under uncertainty is a crucial issue with the recent trend of integration of uncertain renewable sources and mobile electric vehicle loads. The necessary and sufficient condition for small-signal stability is that all eigenvalues of reduced system Jacobian must have a negative real part. This method of stability assessment requires evaluation of all or a set of eigenvalues of reduced Jacobian, which involves a matrix inverse. The reduced Jacobian based stability assessment works for a specified operating state, but becomes inapplicable in the presence of state uncertainties, mostly due to the non-linearity induced from the power flow Jacobian matrix inversion. Some important, relevant works on DAE and stability are presented in \cite{nguyen2017contraction, wolter, Longfamily, zz, dhagash, Molzahn, aolaritei2017distributed, ali2017transversality}. 

The probabilistic small-signal stability (PSSS) assessments have been developed to deal with the issue of stability under uncertainty \cite{burchett1978probabilistic,pans2005probabilistic,arrieta2007contingency}. The core idea is to derive the probability density function ($pdf$) of an uncertain system output, PSSS measure like critical eigenvalue or minimum damping ratio (MDR), based on a known uncertainty of input state \cite{preece2014probabilistic}. The methods of PSSS assessment broadly fall under two categories: i) numerical methods, and ii) analytical methods.

In numerical approaches, the Monte Carlo simulations (MCS) can be employed to determine the $pdf$ of the various stability indices \cite{xu2005probabilistic,rueda2009assessment,preece2014risk}. To improve computational performance, methods such as quasi-Monte Carlo \cite{huang2013quasi} have also been used to assess PSSS. Nevertheless, an essential requirement of a large number of simulations leads to computational burden too high to be implemented for detailed probabilistic studies.

Analytical methods, free from the imposition of parametric output distributions, are also used for PSSS assessment. These methods include probabilistic collocation method (PCM) \cite{preece2012probabilistic}, point estimate methods (PEMs) \cite{rueda2012probabilistic}, analytical cumulant-based approaches \cite{wang2000improved}. A comprehensive comparative study is presented in \cite{preece2014probabilistic}. These methods suffer from issues related to complex formulation and inaccurate first and second-order approximations of sensitivities of eigenvalues. The detailed review of these techniques and approaches can be found in \cite{xu2017probabilistic,wang2016probabilistic,krishan2017analysis,hasan2019existing}. Regardless of the method employed, all these PSSS assessment methods tried to solve specific uncertainty issues considering a typical $pdf$ for uncertainties in the wind, solar, or load. The modeling error in input uncertainty descriptions and effects of assumptions taken get propagated to output $pdf$. Further, these methods do not provide insight into critical eigenvalue behavior in state-space and are mostly limited to one or two-dimensional subspace at a time.

In this paper, we present a novel probabilistic robust small-signal stability (PRS) framework for small-signal stability assessment in a state subspace. This state subspace can be defined based on measurement errors or the level of external disturbances that move the system within such a sub-domain. Further, we define PRS as follows. If all the inner points of a subset of the state-space are small-signal stable with \textit{at least} a given level of confidence, then such a subset is called probabilistic robust stable subspace with respect to such a confidence level. In other words, PRS is concerned that for a given uncertain state-subspace $\mathcal{X}$ in state space, what the probability with that any state point $\mathbf{x} \in \mathcal{X}$ is stable. To the best of our knowledge, this is the first instance when the PRS certificate has been developed and proposed using GP learning for a power system. Different from some existing Monte Carlo-based approaches, our proposed PRS framework does not characterize how many stable operating points in $\mathcal{X}$. Specifically, our goal is to quantify the probability with that any operating point $\mathbf{x} \in \mathcal{X}$ is stable. An advantage of this PRS is that it provides a non-parametric, computationally efficient, and less complex modeling alternative for probabilistic small-signal stability assessment.

More importantly, the proposed PRS framework does not require uncertainty modeling, such as a prior distribution of uncertain inputs, and hence provides a generalized framework for the stability under uncertainty. The main focus of this PRS framework is providing a PRS certificate for a given state subspace being PRS. Along with the certificate, we formulate \textit{Subspace-based Search} and \textit{Confidence-based Search} problems for cases wherein PRS certification can not be verified. The framework is built upon the Gaussian process upper confidence bound (\textit{GP-UCB}) search algorithm \cite{zhai2019region}. The \textit{GP-UCB} is used for sampling the state points inside state subspace $\mathcal{X}$ to learn the behaviors of the critical eigenvalues, which are closest to the imaginary axis for a small-signal stable system. 

The main contributions of the work are as follows: 
\begin{itemize}
    \item Defining probabilistic robust small-signal stability (PRS) and developing a certificate for a subspace to be PRS with a given confidence level. 
    \item Development of a PRS framework with \textit{Subspace-based Search} and \textit{Confidence-based Search} for a state subspace that is not satisfying the PRS certificate with a given probability.
    \item Developing a novel and fast GP learning scheme to learn and analyze the critical eigenvalue behavior in a multi-dimensional state-subspace. 
\end{itemize}
\color{black}

The objective of this work is to present the novel idea of PRS for the power system. We first build up the background in Section \ref{sec:back} by providing a brief review of small-signal stability assessment, GP, and \textit{GP-UCB}methods. From this, we build the PRS framework, which is presented in Section \ref{sec:prs}. This section includes the main results on PRS certificate with details on the \textit{Subspace-based  Search} and \textit{Confidence-based Search} mechanisms. The simulations and discussions of results are provided in Section \ref{sec:dis} while conclusions are drawn in Section \ref{sec:con} with future scope. We use stability to refer to small-signal stability for brevity.

% \color{red}
% This work will develop a robust stability assessment which can verify the system's stability where the operating states are uncertain. We further assume that the operating point lies within a state subspace specified by, for example, the error bound of state estimation. 

% The uncertain operating state makes the symmetric part $L$ also uncertain. We, therefore, propose GP regression for learning the maximum eigenvalue of uncertain $L$ as an unknown function $h(\mathbf{x})$ that is a GP sequentially measured by
% \begin{equation*}
%     q^{(i)} = h(\mathbf{x}^{(i)})+\epsilon, \quad i\in Z^{+}.
% \end{equation*}

% Here $q^{(i)}$ refers to the observed function value for the input $\mathbf{x}^{(i)}$ at the $i$-th sampling, and the state-estimation-related noise $\epsilon$. With the GP approach, we can obtain the posterior distribution over $h(\mathbf{x})$. The main result is the following (see also \cite{zhai2019region}).

\color{black}

\section{Background}\label{sec:back}
The proposed PRS framework in this work has three background building blocks. This section presents these three blocks, namely: Power System Modeling and Stability Assessment, GP and GP-UCB. 

\subsection{Power System Modeling and Stability Assessment}

Consider the power system dynamics that can be expressed as semi-explicit DAE as:
\begin{equation}
\begin{aligned}\label{eq:DAE}
    \mathbf{\dot x} &= f(\mathbf{x},\mathbf{y}) \\
             {0}    &= g(\mathbf{x},\mathbf{y})
\end{aligned} %\vspace{-9mm}
\end{equation}
Here, $\mathbf{x},\, \mathbf{y}$ are dynamic and algebraic variable vectors respectively and $f(\cdot)$ and $g(\cdot)$ are sets of differential and algebraic equations respectively, expressing the power system's behavior. The DAE system can be linearized and expressed as: 
\begin{equation}
\begin{aligned}\label{eq:DAE}
    \delta \mathbf{\dot x} &= A \, \delta \mathbf{x}+ B \, \delta \mathbf{y} \\
             {0}    & = C \,\delta \mathbf{x}+ D \, \delta \mathbf{y}
\end{aligned} %\vspace{-9mm}
\end{equation}

Conventionally, to assess the stability of the DAE system, we rely on the reduced Jacobian matrix obtained by eliminating the algebraic variables $\mathbf{y}$ \cite{kundur1994power}. With an invertible matrix $D$, we have the reduced Jacobian $J_r = A - B D^{-1} C$. The linearized DAE system is small-signal stable at the base operating point $\mathbf{x}^\star$ if and only if the reduced Jacobian evaluated at $\mathbf{x}^\star$ is a Hurwitz stable matrix, i.e., $\max_i \left( \Re(\lambda_i(J_r)) \right) < 0$, where $\lambda_i(J_r)$ be the $i^{th}$ eigenvalue of $J_r$ \cite{pareek2018sufficient}. We use $\lambda_c(\mathbf{x})$ to denote the function describing the real part of the maximum eigenvalue of $J_r(\mathbf{x})$ as the state $\mathbf{x}$ varies. Therefore, the behavior of this function $\lambda_c(\mathbf{x})$ dictates the small-signal stability of the system.

In the power system, some attempts have been made to learn and estimate the movement of critical eigenvalue \cite{yang2007critical,luo2009invariant}. The numerical results of these works indicate that eigenvalues of power system Jacobian are continuous in continuous system state space. An important research direction related to these numerical approaches is based on eigenvalue sensitivities with respect to states \cite{yang2007critical,luo2009invariant}. The downsides of this direction have been discussed in \cite{pareek2018sufficient} and others. This paper, however, does not rely on sensitivities but attempts to learn the function representing the critical eigenvalues when the states change. The learned function can benefit other operating procedures in power systems concerning eigenvalues. 

In order to apply GP, we rely on the fact that the eigenvalues of a matrix are continuous functions of its entries. This property is shown in Theorem 5.2 \cite{serre2010matrices}. The details and discussion on continuity of eigenvalues can also be obtained from Section 5.2.3 of \cite{serre2010matrices}. An intuitive argument can be found in \cite{meyer2000matrix}. The roots of a polynomial equation are shown to depend continuously on coefficients of the polynomial \cite{harris1987shorter}. The eigenvalues are in fact roots of \textit{Characteristic Polynomial} of a matrix $M$ and given as $h(\lambda)=\det\{\lambda I-M\}$, a monic polynomial. Further, the coefficients of $h(\lambda)$ can be expressed in terms of the sum of principal minors of $M$. Each of these principle minors depend on the coefficient $m_{ij}$ of matrix $M$ establishing the continuous dependency of roots of $h(\lambda)$ on entries $m_{ij}$ of matrix $M$. Therefore, roots of \textit{Characteristic Polynomial}, i.e., eigenvalues are a continuous function of entries of the matrix. As the maximum operator over a set of continuous values is continuous, the critical eigenvalue of the matrix $J_r$ is continuous with variations in matrix entries. 

It is essential to note the following. In the present work, we are interested in exploring the neighborhood subspace of an operating point. This subspace can be taken as continuous as the abrupt state changes are not under consideration of this work. In various power system operations and conditions, the state-space cannot be considered as continuous. These situations arrive mainly during discrete changes such as contingencies or instances of controller saturation. Identification of continuous subspace is critical in such cases before applying the proposed method.

The reduced Jacobian is a non-linear function of states $J_r(\mathbf{x})$, and so is critical eigenvalue $\lambda_c(\mathbf{x})$. Due to involvement of the matrix inverse, obtaining analytical expression for $\lambda_c(\mathbf{x})$  is not possible. Further, in the situation where the state vector $\mathbf{x}$ is uncertain and can vary in a state subspace $\mathcal{X}$, the $\lambda_c(\mathbf{x})$ will be uncertain and more difficult to estimate. Therefore, we propose this GP learning based method to learn the behavior of $\lambda_c(\mathbf{x}) $ with $\mathbf{x} \in \mathcal{X}$.

\subsection{Gaussian Process Regression}
In the Bayesian optimization paradigm, where the exact objective function is not known, the Gaussian process is used extensively as a modeling tool. The GP is developed as an extension of multi-variate Gaussian and can be considered as a distribution over random functions \cite{williams2006gaussian}.  The Gaussian process is a non-parametric method, hence suitable as a modeling method for probability distributions over functions. 

In power systems, the use of GP has been restricted to the forecasting applications for power and load. The GP has been applied for wind power forecast \cite{lee2013short,chen2013wind,yan2015hybrid}, solar power forecast \cite{sheng2017short}, and electricity demand forecast \cite{blum2013electricity,yang2018power,van2018probabilistic}. Other then these forecasting works, the idea of using GP to learn the dynamics and stability index behavior has not been explored by the power system research community.

First, we define a general framework for GP regression. Let, a training data set $\mathcal{D}=\{\mathbf{x}^{(i)},\hat{\lambda}_c (\mathbf{x}^{(i)})\}^m_{i=1}$ where $\hat{\lambda}_c (\mathbf{x}^{(i)})$ is the observed function value for input vector $\mathbf{x}^{(i)} \in \mathbb{R}^n$ at the $i^{th}$ step. Then the GP regression model can be given as \cite{williams2006gaussian}:
\begin{align}\label{eq:reg}
        \hat{\lambda}_c(\mathbf{x}^{(i)}) = \lambda_c(\mathbf{x}^{(i)})+\epsilon^{(i)}, \quad i=1 \dots m
\end{align}
Here, $\epsilon^{(i)}$ are independent and identically distributed noise variable with zero-mean, $\sigma_n$ standard deviation normal distribution $\left ( \mathcal{N}(0,\sigma_n^2) \right)$. Interested reader can look into \cite{williams2006gaussian} for details of GP fundamentals.

In this work, the unknown critical eigenvalue function is approximated by GP regression. The covariance or kernel function $k(\mathbf{x},\mathbf{x}')$, brings our understanding and assumptions about the function into GP.  A sample set $\hat{\lambda}_{c_m}=[\hat{\lambda}_c(\mathbf{x}^{(1)}), \dots, \hat{\lambda}_c(\mathbf{x}^{(m)})]^T$ at operating points $\mathcal{D}_m = \{\mathbf{x}^{(1)}, \dots ,\mathbf{x}^{(m)}\}$ with Gaussian noise $\epsilon $, and the analytic formula set can be obtained for posterior distribution corresponding to (\ref{eq:reg}). 
\begin{subequations}
\begin{align}
    \mu_m(\mathbf{x}) &= k_m(\mathbf{x})^T(K_m+\sigma^2_n I)^{-1}q_m\\
    k_m(\mathbf{x},\mathbf{x}') &= k(\mathbf{x},\mathbf{x}')-k_m(\mathbf{x})^T(K_m+\sigma^2_n I)^{-1}k_m(\mathbf{x'})\\
    \sigma^2_m(\mathbf{x}) &=  k_m(\mathbf{x},\mathbf{x})
\end{align}
\end{subequations}

Here,  $\mu_m(\mathbf{x})$ is mean,  $k_m(\mathbf{x},\mathbf{x}')$ is covariance and variance is indicated by $\sigma^2_m(\mathbf{x})$. The $\mathbf{x}$ and $\mathbf{x}'$  are two sample operating points from set $\mathcal{D}_m$. The $k_m(\mathbf{x}) = [k(\mathbf{x}^{(1)},\mathbf{x}), \dots,\mathbf{x}^{(m)},\mathbf{x}) ]$ and $K_m=[k(\mathbf{x},\mathbf{x}')]$. In this work, we infuse our prior knowledge about critical eigenvalue function as the squared exponential covariance function with zero mean and unit characteristic length.

\subsection{GP-UCB}
The Gaussian process upper confidence bound (GP-UCB) algorithm is an intuitive Bayesian method \cite{srinivas2009gaussian} for sampling. For a given $\delta \in (0,1)$, in a state subspace $\mathcal{X}$, our goal is to learn the mean $\mu(\mathbf{x})$ for critical eigenvalue function $\lambda_{c}(\mathbf{x})$ with least standard deviation $\sigma(\mathbf{x})$ and confidence level $1-\delta$. A combined strategy to strike the balance between exploration and exploitation can be used for sampling the point $\mathbf{x}^{(i)}$. With $\beta_i$ taken independent of state vector the sampling strategy will be: 
\begin{align}\label{eq:gpucb}
    \mathbf{x}^{(i)} = \argmax_{\mathbf{x}^{(i)} \in \mathcal{X}}  \left \{ \mu_{i}(\mathbf{x}) + \sqrt{\beta_{i+1}}\sigma_{i}(\mathbf{x}) \right \}
\end{align} 

Here, (\ref{eq:gpucb}) suggests that $\mathbf{x}^{i}$ is selected where $\mu_{i}(\mathbf{x}) + \beta_{i+1}^{1/2}\sigma_{i}(\mathbf{x})$ maximizes. The $\mu_{i}(\mathbf{x})$ contributes to enlarging the level set of critical eigenvalue while $\sigma_{i}(\mathbf{x})$ helps in minimizing the uncertainty. Interested reader can obtain the details of this sampling strategy and \textit{GP-UCB} from \cite{srinivas2009gaussian}.

\section{Main Result: PRS Framework}\label{sec:prs}
In this section, we will present the main result for developing the PRS framework. Prior to this result, we note that various works have been proposed under the generic name \textit{probabilistic robust control} to perform robust control under stochastic uncertainties \cite{calafiore2006scenario}. The system performance is said to be robustly satisfied (with a fixed probability) if a guarantee is provided against \textit{almost all} of the possible uncertainties \cite{calafiore2007probabilistic}. A similar analogy in the context of power system small-signal stability can be used here for the PRS framework. 

Given $\delta \in (0,1)$, if all the points in a subspace are stable with the probability of \textit{at least} $1-\delta$, then that subspace is called PRS subspace. For developing the basis of PRS certificate, we first present a result on regret bound, i.e., error bound, in Theorem \ref{theo:main} for \textit{GP-UCB} represented in (\ref{eq:gpucb}). 

\begin{theorem}\label{theo:main}
Let $\lambda_c(\mathbf{x})$ be the critical eigenvalue function in state $\mathbf{x}$ with the noise $\epsilon$ bounded by $\sigma_n$. Then the following holds with the probability of at least $1-\delta$ with $\delta\in(0,1)$
\begin{align}\label{eq:rbound}
    |\lambda_c(\mathbf{x})-\mu_{m}(\mathbf{x})|\leq \sqrt{\beta_{m+1}} \, \sigma_{m}(\mathbf{x}), \quad \forall \mathbf{x}\in \mathcal{X}.
\end{align}
Here $\mathcal{X}$ is the sample space wherein the states $\mathbf{x}$ lie, and $m$ is the number of sampling points.  $\beta_m=2\|V\|_k^2+300\gamma_m\ln^3{(m/\delta)}$ defined in \cite{zhai2019region}.
\end{theorem}
\begin{proof}
The results follow directly from Theorem 6 in \cite{srinivas2012information} and Theorem 1 in \cite{zhai2019region}. 
\end{proof}

Most importantly, following Theorem \ref{theo:main}, we obtain a probabilistic robust stability certificate as the following. 

\begin{theorem} {(PRS certificate)} \label{theo:prsc}
For a given $\delta \in (0,1)$, the linearized DAE system (\ref{eq:DAE}) is probabilistic robust small-signal stable (PRS) in a state subspace $\mathcal{X}$ with probability $1-\delta$ if
\begin{equation}
    p_m(\mathbf{x,\delta}) < 0
\end{equation}
where $p_m(\mathbf{x,\delta}) = \max_{\mathbf{x}\in \mathcal{X}} \left \{ \mu_m(\mathbf{x}) +  \sqrt{\beta_{m+1}} \, \sigma_m(\mathbf{x}) \right \}$ with $m$ sampling points.
\end{theorem}
\begin{proof}
The definition of PRS certificate $ p_m(\mathbf{x,\delta}) < 0$ can be obtained directly from Theorem \ref{theo:main}. 
\end{proof}

The PRS certificate $ p_m(\mathbf{x,\delta}) < 0 $ verifies that with the probability $1-\delta$ any state vector $\mathbf{x}$ in the given state subspace $\mathcal{X}$ is small-signal stable. This idea does not quantify the number of points that are stable in a subspace but gives the probability with that any inner point is stable in the considered subspace. 

Based upon Theorem \ref{theo:prsc}, we continue the PRS framework, which covers various conditions as:
\begin{enumerate}
    \item \textit{If the PRS certificate holds or} $p_m(\mathbf{x,\delta})< 0$:\\
    The state subspace $\mathcal{X}$ is probabilistic robust stable with the confidence level of $1-\delta$.
    \item \textit{Otherwise}: 
    \begin{enumerate}
       \item \textbf{Subspace-based  Search:} \\
       Find $\mathcal{X'} \subsetneq \mathcal{X}$ such that $ p_m(\mathbf{x,\delta})< 0 $, $ \forall \, \mathbf{x} \in \mathcal{X'} $.
       \item \textbf{Confidence-based Search:} \\
       Find $\delta' > \delta$ such that $ p_m(\mathbf{x,\delta'})< 0$, $ \forall \, \mathbf{x} \in \mathcal{X} $.
    \end{enumerate}
\end{enumerate}

The rationale behind the second step of the PRS framework is that in case the PRS certification cannot be established for a given set $\mathcal{X}$, one can proceed with two alternative searches. \textit{Subspace-based Search} results into a subspace $\mathcal{X'}$, of state space, which is PRS with probability $1-\delta$. This $\mathcal{X'}$ is also a subspace of the original state subspace $\mathcal{X}$. The other alternative \textit{Confidence-based Search} results in a lower confidence level $1-\delta'$, by which the original subspace $\mathcal{X}$ is PRS. Therefore, with the PRS certificate and two alternative searches, the PRS framework is complete as it caters to all possible scenarios.  

\begin{remark}
The existence of $\mathcal{X'}$  in \textit{Subspace-based Search} is guaranteed if $\exists \, \mathbf{x}^\star \in \mathcal{X}$ such that $\max_i \left( \Re(\lambda_i(J_r(\mathbf{x}^\star ))) \right) < 0$. In many cases, $\mathbf{x}^\star$ is selected as the base stable operating point around that the subspace $\mathcal{X}$ is constructed. Therefore, we can at least find $\mathcal{X'}$ containing the stable base point.
\end{remark}

For \textit{Confidence-based Search}, by increasing $\delta$ to $\delta'$, we obtain a lower confidence level $1-\delta'$. However, we do not guarantee that the new confidence level $1-\delta'$ can be large enough for any practical meaning. 

Within our PRS framework, it is important to determine when the \textit{GP-UCB} search can be terminated or the search is completed. As the \textit{GP-UCB} attempts to minimize the learning uncertainty related to $\sigma_m(\mathbf{x})$ in \eqref{eq:gpucb}, the change in values of $\sigma_m(\mathbf{x})$ indicates the level of learning errors with increasing number of training samples $m$. Further, this learning level will also be reflected in variation of $p_m(\mathbf{x,\delta})$ while increasing $m$. Therefore, either of the two indicators, i.e., $\sigma_m(\mathbf{x})$ and $p_m(\mathbf{x,\delta})$, can be used to decide whether the \textit{GP-UCB} search for learning $\lambda_c(\mathbf{x})$ completes. Moreover, we propose to use the real part of critical eigenvalue as a stability index instead of minimum damping ratio (MDR) in this work. Although, the proposed framework can be extended to accommodate any stability index such as MDR. 
 
% \subsection{Effects of confidence interval ($1-\delta$)}
% On the number of points require to reach a particular level of $\sigma$. 
\color{black}
\section{Simulations and Discussions}\label{sec:dis}

In this work, we have used IEEE three-machine nine-bus system \cite{pai2012energy} for testing and validation of the proposed PRS framework. The system has three $PV$ buses with conventional generators at bus number $1$, $2$, and $3$. The generator $1$ is considered as slack bus and has the highest inertia constant. We construct the reduced system Jacobian ($J_r$) and test PRS certificate and framework for various $PV$ bus subspaces. As the PRS framework, in Section \ref{sec:prs},  is developed for a general \textit{n-dimensional} space vector, the extension to the load bus subspace is straight forward. We use $|V_i|$ to indicate the node voltage magnitude at $i^{th}$ bus, while $P_{g_i}$ and $Q_{g_i}$ to indicate real and reactive power output of generator connected to $i^{th}$ bus. All numerical values are in per unit ($pu$). A general set describing a subspace is indicated with $\mathcal{X}$ while a set for which the PRS certificate (Theorem \ref{theo:prsc}) can be satisfied is indicated with $\mathcal{X}_c$. The values of the small-signal stable base point are given in Appendix Table \ref{tab:base}. 

\begin{figure}[t]
    \centering
    \includegraphics[width=\columnwidth]{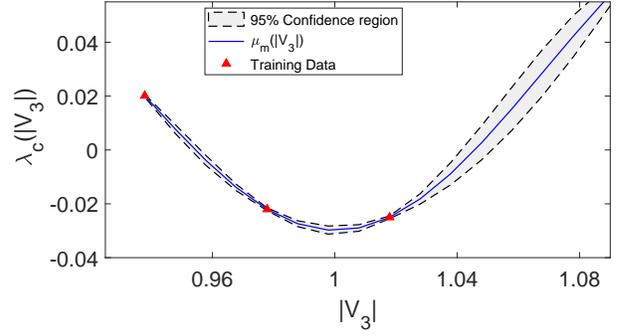}
    \caption{An illustration on \textit{GP-UCB} working on 1-D state subspace of $|V_3|$. The blue
line denotes mean values of the unknown function $\lambda_c(|V_3|)$, and the gray shade describes the 95\% confidence interval. The three red triangles represent three different sampling points.}
    \label{fig:gplearning}
\end{figure}
\begin{figure}[t]
    \centering
    \includegraphics[width=\columnwidth]{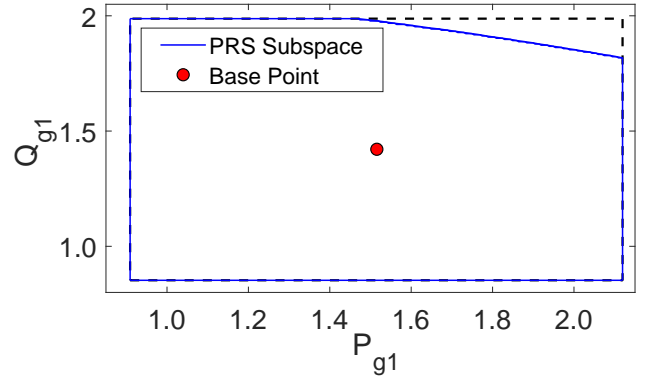}
    \caption{PRS certificate testing for $\mathcal{X}_1=\{ 0.909 \leq P_{g_1}\leq 2.119;\, 0.852 \leq Q_{g_1} \leq 1.988 \}$ indicated with black dashed rectangle}
    \label{fig:PgQg1}
\end{figure}

First, we discuss the GP learning process in \textit{one - dimensional} voltage magnitude $|V_3|$ state subspace. Figure \ref{fig:gplearning} shows the learning mechanism for $\lambda_c$ as a function of $|V_3|$ with three sampling points shown as red triangles ($m=3$). The mean $\mu_m(|V_3|)$ is shown in blue and $\pm\, 2\,\sigma_m(|V_3|)$ is indicated with gray area covering $95\%$ confidence interval. It is clear that for $|V_3| \geq 1.04$, the uncertainty is higher, leading to a larger gray area for achieving the $95\%$ confidence level. Therefore, it can be concluded that with a sufficiently large number of sampling points, the uncertainty in learning $\lambda_c(\mathbf{x})$ can be decreased to an acceptable level. Further, the objective of the PRS framework is to certify a neighborhood state subspace of a stable base point. It is not tasked to certify or evaluate the stability in the entire state space. This allows minimizing the uncertainty values efficiently with less number of sampling points.

\begin{figure}[t]
    \centering
    \includegraphics[width=\columnwidth]{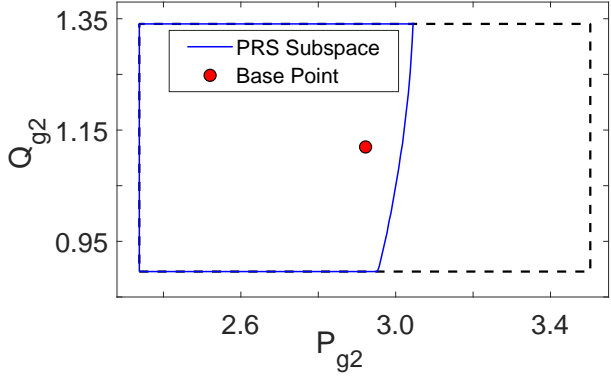}
    \caption{PRS certificate testing for $\mathcal{X}_2=\{ 2.338 \leq P_{g_2}\leq 3.503;\, 0.895 \leq Q_{g_2} \leq 1.341 \}$ indicated with black dashed rectangle}
    \label{fig:PgQg2}
\end{figure}

For power system stability, the state space with real and reactive power variables is crucial as they reflect the generator dynamics impact directly. Further, different generators have a different level of inertia, thus leading to a different stable subspace around the base point. To study PRS with real and reactive power variations, the PRS certificate verification is done for each generator for different subspace. The Figures \ref{fig:PgQg1}, \ref{fig:PgQg2} and \ref{fig:PgQg3} shows results for PRS testing in power subspace corresponding to generators $1$, $2$ and $3$ respectively. The black rectangle indicates the region covered by different subspace set $\mathcal{X}_1$, $\mathcal{X}_2$ and $\mathcal{X}_3$ while the region covered by blue line is subspace for which $\mu_m(\mathbf{x})\, + \, 2 \,\sigma_m(\mathbf{x}) \leq 0$. All these figures highlight the difference in the PRS subspace of different generators. The higher area covered by the blue line in Figure \ref{fig:PgQg1} is indicative of the fact that generator $1$ can handle higher variations in power before leading into instability. As $p_m(\mathbf{x,\delta}) \nless 0$ for sets  $\mathcal{X}_1$, $\mathcal{X}_2$ and $\mathcal{X}_3$, we can perform \textit{Subspace-based  Search} to find a PRS subspace.

The Figure \ref{fig:p1v1} shows two different subspace where PRS certificate validation is performed in $P_{g_1}-|V_1|$ subspace. Further, it shows a subspace in blue dash line for which $p_m(\mathbf{x,\delta}) < 0$. Therefore, every state point inside the blue rectangle has a $95\%$ chance of being stable. This PRS subspace is obtained using \textit{Subspace-based  Search}. Here, we use a subspace description with rectangular geometry for conveying the main idea, although any set description can be used with the proposed PRS framework. Figure \ref{fig:p1v1} also suggests that a maximum PRS subspace can also be found, which will have a lesser range in $|V_1|$ and higher in $P_{g_1}$ state dimension. This result is also indicative of the fact that the PRS subspace is non-unique, and there exist multiple such subspaces in state space varying in descriptions. Therefore, \textit{Subspace-based  Search} can be tasked to find the maximum range in the dimension which is most uncertain or exhibit larger variations than other states. This paper is focused on presenting the PRS framework idea; therefore, such customization is not in scope.

\begin{figure}[t]
    \centering
    \includegraphics[width=\columnwidth]{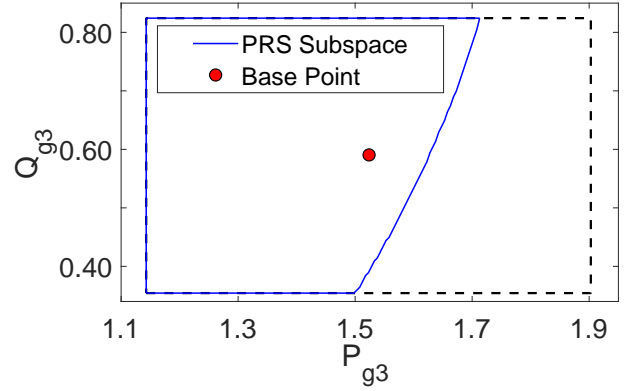}
    \caption{PRS certificate testing for $\mathcal{X}_3=\{ 1.143 \leq P_{g_3}\leq 1.914;\, 0.354 \leq Q_{g_3} \leq 0.824 \}$ indicated with black dashed rectangle}
    \label{fig:PgQg3}
\end{figure}

\begin{figure}[t]
    \centering
    \includegraphics[width=0.95\columnwidth]{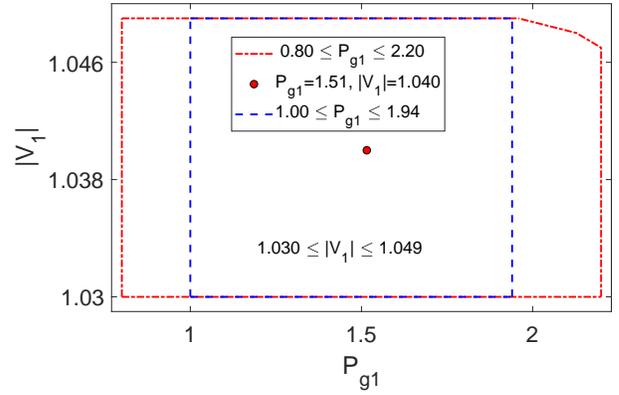}
    \caption{$P_{g_1}-|V_1|$ subspace $\mathcal{X}$ such that $p_m(\mathbf{x,\delta}) \nless 0 \, \forall \,  \mathbf{x} \, \in \mathcal{X} $ with red line and subspace $\mathcal{X}_c \subsetneq \mathcal{X}$ such that $ p_m(\mathbf{x,\delta}) < 0 \, \forall \,  \mathbf{x} \, \in \mathcal{X}_c $ with blue line}
    \label{fig:p1v1}
\end{figure}

Similar to the result shown in Figure \ref{fig:p1v1}, the Table \ref{tab:prsc} contains the dimensions of different $\mathcal{X}_c$ for different variables. We observed that $|V|$ is the limiting variable and has the least acceptable variations in comparison to the power subspace. Here, the objective is again to show that the PRS framework has been able to obtain meaningful $\mathcal{X}_c$ in different dimensions and variable subspace and $\mathcal{X}_c$ is not indicating largest possible subspace in Table \ref{tab:prsc}. 
\begin{table}[t]
  \centering
  \caption{PRS Certified Subspace $\mathcal{X}_c$ Type and Dimensions in $pu$}
  \bgroup
\def\arraystretch{1.2}
    \begin{tabular}{c|cccc}
      Space Type    & Variables & Minimum & {Maximum} & {$\Delta$} \\
          \hline
     \multirow{3}{*}{\shortstack{$|V|$ Space  \\ 3-Dimensional }} & $|V_1|$    & 1.032 & 1.048 & 0.017 \\
          & $|V_2|$    & 1.017 & 1.033 & 0.016\\
          & $|V_3|$    & 1.017 & 1.033 & 0.016 \\
    \hline
    \multirow{3}{*}{\shortstack{$P_g$ Space  \\ 3-Dimensional }} & $P_{g_1}$    & 1.478 & 1.554 & 0.076 \\
          & $P_{g_2}$    & 2.849 & 2.995 & 0.146 \\
          & $P_{g_3}$    & 1.486 & 1.562 & 0.076 \\
    \hline
    \multirow{6}{*}{\shortstack{$P_g-Q_g$ Space  \\ 6-Dimensional }}& $P_{g_1}$    & 1.478 & 1.554 & 0.076\\
          & $P_{g_2}$    & 2.849 & 2.995 & 0.146 \\
          & $P_{g_3}$    & 1.486 & 1.562 & 0.076 \\
          & $Q_{g_1}$    & 1.385 & 1.457 & 0.071 \\
          & $Q_{g_2}$    & 1.092 & 1.148 & 0.056 \\
          & $Q_{g_3}$    & 0.576 & 0.605 & 0.030 \\
    \hline
    \multirow{9}{*}{\shortstack{$P_g-Q_g-|V|$ Space  \\ 9-Dimensional }} & $P_{g_1}$    & 1.478 & 1.554 & 0.076 \\
          & $P_{g_2}$    & 2.849 & 2.995 & 0.146 \\
          & $P_{g_3}$    & 1.486 & 1.562 & 0.076 \\
          & $Q_{g_1}$    & 1.385 & 1.457 & 0.071 \\
          & $Q_{g_2}$    & 1.092 & 1.148 & 0.056 \\
          & $Q_{g_3}$    & 0.576 & 0.605 & 0.030 \\
          & $|V_1|$    & 1.035 & 1.045 & 0.010 \\
          & $|V_2|$    & 1.020 & 1.030 & 0.010 \\
          & $|V_3|$    & 1.020 & 1.030 & 0.010 \\
      \hline
    \end{tabular}%
    \egroup
  \label{tab:prsc}%
\end{table}%

The proposed GP learning method can be used to understand the behaviour of $\lambda_c(\mathbf{x})$. The Figure \ref{fig:v2v3} depicts variations in $\mu_m(\mathbf{x}) +  \sqrt{\beta_{m+1}} \sigma_m(\mathbf{x})$ in $|V_2|-|V_3|$ subspace around the base point. From Theorem \ref{theo:main}, it is clear that Figure \ref{fig:v2v3} represents upper bound of critical eigenvalue with probability $1-\delta$. This upper bound, in this subspace, is smooth and an edge near $|V_2|=|V_3|=1.05$ cross the stability boundary. Further, from Theorem \ref{theo:main} it can be concluded that if learning uncertainty is very less i.e. $\max \{\sigma_m(\mathbf{x})\} \lll 1 \, \forall \, \mathbf{x} \in \mathcal{X}$, then even with higher value of $\beta_{m+1}$ the difference between actual $\lambda_c(\mathbf{x})$ and its estimated mean $\mu_m(\mathbf{x})$ will be very small. The Figure \ref{fig:v2v3} indicate one such situation with $\max \left\{ \sigma_m(\mathbf{x})\right\} \sim 10^{-8}$. Therefore, in this situation, the plane shown very closely represent $\lambda_c(\mathbf{x})$. 

\begin{figure}[t]
    \centering
    \includegraphics[width= 0.8 \columnwidth]{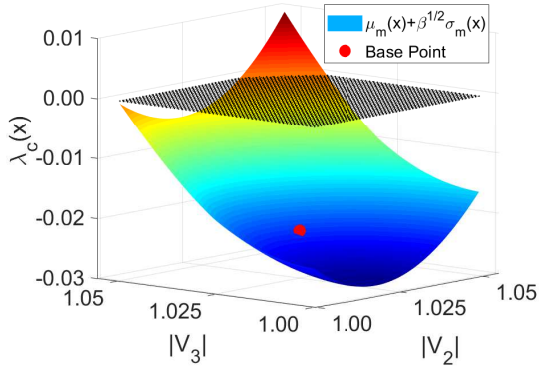}
    \caption{Critical eigenvalue ($\lambda_c$) upper bound plane in \textit{two-dimensional} $|V_2|-|V_3|$ subspace with $95\%$ confidence level, with $\max \left\{ \sigma_m(\mathbf{x})\right\} \sim 10^{-8}$}
    \label{fig:v2v3}
\end{figure}

\subsubsection*{Time Consumption Report}
All simulations in this work are performed using GPML \cite{rasmussen2010gaussian} with MATLAB 2018b on a machine with Intel Xeon E5-1630v4 having 3.70 GHz clock speed and 16.0 GB of RAM. It takes approximately $2.3~sec$ to obtain results of Figure \ref{fig:v2v3} while PRS certificate is verified in approximately $2.5~sec$ in Figure \ref{fig:p1v1}. The time consumption depends on the initial sampling point and can be improved using better initialization strategies, which are not in the scope of this work. Further, for higher dimensions and larger data sets, sparse Gaussian processes \cite{snelson2006sparse} can be used for less time consumption. 

\section{Conclusion}\label{sec:con}
In this paper, we present a novel probabilistic robust small-signal stability (PRS) framework based on the non-parametric Gaussian process to learn the behavior of the critical eigenvalue of the reduced Jacobian. In particular, we use GP to estimate the real-part of the critical eigenvalue by learning its upper bound and thus assess the system stability. The robustness of system stability is verified in the state subspace with a given probability by establishing a PRS certificate. In case the PRS certificate is not satisfied, we propose \textit{Subspace-based  Search} and \textit{Confidence-based Search} to complete PRS framework. The simulation results on the WSCC network illustrates the performance of the proposed method in certifying the PRS certificate and learning the critical eigenvalues.

The present work opens up a new dimension in the area of probabilistic stability assessment of the power system. The certificate presented here can be applied to larger test systems and can be further optimized, for example, to be less time-consuming. The learning of the critical eigenvalue behavior can provide a refined understanding of the stability under uncertainty. There also exists a potential for solving problems to find the largest PRS subspace for a given confidence level. Based on this work, one can design a general PSSS assessment method that neither require the $pdf$ of input uncertainty nor rely on state-specific approximations. These ideas will be explored further in upcoming works. 

\section*{Acknowledgement}
We thank Prof. Ashu Verma and Dr. Zhai Chao for their suggestions.

\section*{Appendix}

\begin{table}[h]
  \centering
  \caption{Small-Signal Stable Base Point}
    \bgroup
\def\arraystretch{1.2}
    \begin{tabular}{c|ccc}
          & \multicolumn{3}{c}{Bus Number} \\
          & \multicolumn{1}{c}{1} & \multicolumn{1}{c}{2} & \multicolumn{1}{c}{3} \\
    \hline
    Real Power $( P_g)$ & 1.515 & 2.922 & 1.523 \\
    Reactive Power $(Q_g)$ & 1.421 & 1.119 & 0.590 \\
    Voltage Magnitude $(|V|)$ & 1.040  & 1.025 & 1.025 \\
    \hline
    \end{tabular}%
    \egroup
  \label{tab:base}%
\end{table}%

\bibliographystyle{IEEEtran}
\bibliography{main}

\end{document}